\documentclass{llncs}
\usepackage{graphicx,amssymb,amsmath}
\usepackage[utf8]{inputenc}
\usepackage{amssymb}
\newcommand{\remove}[1]{}
\usepackage{graphicx}
\usepackage{url}
\usepackage{textcomp}
\usepackage{amsfonts,amsmath}
\usepackage{framed}
\usepackage{boxedminipage}
\usepackage{graphicx}
\usepackage{verbatim}
\usepackage{amsmath}
\usepackage{amssymb}

\usepackage[lined,boxed,commentsnumbered,ruled,vlined,noend,linesnumbered,boxed]
{algorithm2e}

\newtheorem{observation}{Observation}

\begin{document}

\title{Maximum-width Axis-Parallel Empty Rectangular Annulus} 
\author{Arpita Baral\inst{1},
Abhilash Gondane\inst{2},
Sanjib Sadhu\inst{2},
Priya Ranjan Sinha Mahapatra\inst{1},
\institute{University of Kalyani, India \\
\and
National Institute of Technology, Durgapur, India \\
}
}
\maketitle

\begin{abstract}
Given a set $P$ of $n$ points on $\mathbb R^{2}$, we address the problem of computing an axis-parallel empty rectangular  
annulus $A$ of maximum-width such that no point of $P$ lies inside $A$ but all points of $P$ must lie 
inside, outside and on the boundaries of two parallel rectangles forming the annulus $A$. 
We propose an $O(n^3)$ time and $O(n)$ space algorithm to solve the problem. In a particular case when the inner rectangle 
of an axis-parallel empty rectangular annulus reduces to an input  point we can solve the problem in $O(n \log n)$ time and $O(n)$ space.
\end{abstract}
\section{Introduction}
A set of $n$ points $P$ on $\mathbb R^{2}$ is said to be {\sf enclosed} by a geometric (or enclosing) object $C$ if 
all points of $P$ must lie inside $C$ and on the boundary of $C$. The problem of enclosing the input point set $P$ 
using a {\sf minimum sized} geometric object $C$ such as a circle~\cite{ps-cgi-90}, rectangle~\cite{tou-sgprc-83}, 
triangle~\cite{oamb-oafmet-86}, circular annulus~\cite{w-nmpp-86,rz-epccmrsare-92,efnn-rauvd-89,ast-apsgo-94,as-eago-98}, 
rectilinear annulus~\cite{ght-oafepranw-09}, rectangular annulus~\cite{jmkd-mwra-2012} etc has been extensively studied 
in computational geometry over the last few decades. Here we start the discussion with the enclosing problem that uses various annulus 
as an enclosing object. Among various types of annulus, the enclosing problem using circular annulus has been studied 
extensively~\cite{w-nmpp-86,rz-epccmrsare-92,efnn-rauvd-89,ast-apsgo-94,as-eago-98}. The objective of this problem is 
to find a circular annulus of minimum-width that encloses $P$. Here circular annulus region is formed by two concentric 
circles. Gluchshenkoa et al.~\cite{ght-oafepranw-09} considered 
the problem of finding a rectilinear annulus of minimum-width which encloses $P$. For this problem, the annulus region 
is formed by two concentric axis-parallel squares. Recently, Bae~\cite{bae-cmwea-2017} studied this square
annulus problem in arbitrary orientation where annulus is the open region between two concentric squares. 
Mukherjee et al.~\cite{jmkd-mwra-2012} considered the problem of 
identifying a rectangular annulus of minimum-width which encloses $P$. In this problem, it is interesting to note 
that two mutually parallel rectangles forming the annulus region, are not necessary to be {\sf concentric}. 
Moreover the 
orientation of such rectangles is not restricted to be axis-parallel. Further details on various annulus problem can be found
in~\cite{bae-cmwea-2017,bg-opapp-14,bbdg-opap-98,bbbrw-ccmwaps-98,AHIMPR-03,DuncanGR-97} and the references therein.\\

As per we are aware, there has been little work on
finding empty annulus of maximum-width for an input point set $P$. D{\'{\i}}az{-}B{\'{a}}{\~{n}}ez et al.~\cite{dhmrs-leap-03} 
first 
studied the problem of finding an empty circular annulus of maximum-width and proposed $O(n^3\log n)$ time and $O(n)$ space 
algorithm to solve it. Mahapatra~\cite{mahapatra-larea-2012} considered the problem of identifying an axis-parallel empty 
rectangular annulus of maximum-width for the point set $P$ and proposed an {\sf incorrect} $O(n^2)$ time algorithm to solve 
it. Given a point set $P$, note that, for an axis-parallel minimum-width rectangular annulus which encloses $P$, the outer 
or larger rectangle is always the minimum enclosing rectangle of $P$~\cite{jmkd-mwra-2012}. This observation leads to 
develop an $O(n)$ time algorithm to find an axis-parallel rectangular annulus of minimum-width which encloses
$P$~\cite{jmkd-mwra-2012}. However for the empty axis-parallel rectangular annulus problem of maximum-width, the number of 
potential outer rectangles forming an empty rectangular annulus is $O(n^4)$. This implies that $O(n^5)$ algorithm can be
developed to solve this empty annulus problem using the result in
~\cite{jmkd-mwra-2012}. Here we propose an $O(n^3)$ time and $O(n)$ space algorithm for finding an axis-parallel
empty rectangular annulus of maximum-width for a given point set $P$. Note that the problem of axis-parallel empty 
rectangular annulus of maximum-width is equivalent to the problem when the empty annulus region is generated by two 
concentric rectangles.

The paper is organized as follows: In Section~\ref{problemdefandter} we discuss the problem of identifying
an axis-parallel empty rectangular annulus of maximum-width after introducing some notations. In Section~\ref{proposedalgorithm}
we describe our new algorithm and prove its correctness. Section~\ref{conclusion} concludes the paper.

\section{Problem definition and terminologies}\label{problemdefandter}
We begin by introducing some notations. Let $P = \{p_1, \ldots, p_n\}$ be a set of $n$ points on $\mathbb R^{2}$. 
Let the $x$ and $y$-coordinate of a point $p_i$ be denoted as $x(i)$ and $y(i)$ respectively.
Two axis-parallel rectangles $R$ and $R'$ are said to be parallel to each other when one of the sides of rectangle $R$
is parallel to a side of $R'$. 
Let $R_{in}$ and $R_{out}$ be two axis parallel rectangles such that $R_{in} \subset R_{out}$.
The rectangular annulus $A$ formed by two such axis-parallel rectangles $R_{in}$ and $R_{out}$ is the
{\sf open region} between $R_{in}$ and $R_{out}$ where $R_{in}$ has non-zero area. See Fig. \ref{fig1:Cases} 
for a demonstration. 
We use the term inner (resp. outer) for the 
smaller (resp. larger) rectangle of rectangular annulus $A$. In this paper the rectangle will always imply
an axis parallel rectangle. The {\sf top-width} of the rectangular annulus $A$ is 
the perpendicular distance between top sides of its inner and outer rectangles. Similarly we define the  
{\sf bottom-width}, {\sf right-width} and {\sf left-width} of $A$. 
The minimum width among the top-, right-, bottom- and left-widths of a rectangular annulus $A$ is defined as the {\sf width}
of $A$ and is denoted 
by $W(A)$. A rectangular annulus $A$ formed by rectangles $R_{in}$ and $R_{out}$ (
$R_{in} \subset R_{out}$) is said to be {\sf empty} if the following two conditions are satisfied.
\begin{itemize}
 \item [$(i)$] No points of $P$ lie inside $A$.
 \item [$(ii)$] All points of $P$ lie inside the rectangle $R_{in}$ and outside the rectangle $R_{out}$. The input points 
 may lie on the boundaries of both $R_{in}$ and $R_{out}$.
\end{itemize} 
The objective of our problem is to compute an axis-parallel empty rectangular annulus of {\sf maximum-width}
from the given point set $P$. Note that the solution of this problem is not unique.
From now onwards the term annulus is used to mean an axis-parallel empty rectangular annulus.

\section{Proposed Algorithm}\label{proposedalgorithm}
An annulus is defined by its eight edges (Four edges of outer rectangle and four edges of inner rectangle). Each edge of an
annulus passes through a point $p \in P$. See Fig. \ref{fig1:Cases} as an illustration. \\

\begin{figure}[]
 \centering
 \includegraphics[scale=0.5]{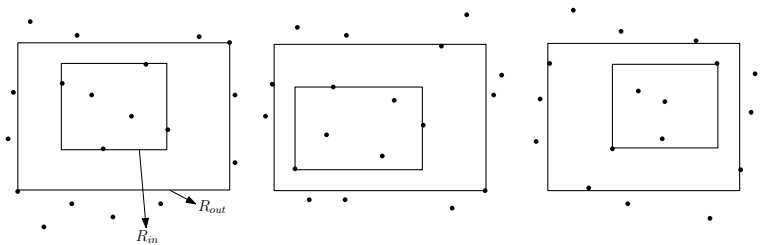}
 \caption{\small{Different configurations of empty annulus.}}  
 \label{fig1:Cases}
 \end{figure}

Initially sort $n$ points of $P$ in ascending order on the basis of their x-coordinates and in descending order on the 
basis of their y-coordinates. Throughout the paper we have assumed that all points are in general position i.e. no 
horizontal or vertical line pass through two points. Two horizontal
lines $Top_{out}$ and $Bot_{out}$ sweeps vertically from top to bottom over the plane and these two lines denote the current 
positions of the top and bottom sides of the outer rectangle defining an empty annulus. Depending on the position of 
$Top_{out}$ and $Bot_{out}$ a horizontal strip is defined as follows.\\

\noindent {\bfseries Definition 1}  A horizontal strip $S(a,b)$ is defined as the open region bounded by two parallel
lines $Top_{out}$ and $Bot_{out}$ where the lines $Top_{out}$, $Bot_{out}$ pass through the 
points $a$ and $b$ respectively, having $y(a) > y(b)$ and $a, b \in P$.\\

\noindent {\bfseries Definition 2} $E(a,b)$ is the set of all empty annuli in $S(a,b)$ such that the top ($Top_{out}$) 
and bottom ($Bot_{out}$) edges of outer rectangle of any annulus $A \in E(a,b)$ pass through the points 
$a$ and $b$ respectively.\\

We now state the following simple observation. 

\begin{observation}\cite{jmkd-mwra-2012}\label{obs1}
Given an outer rectangle $R_{out}$ generated from the point set $P$ on $\mathbb R^{2}$, the empty annulus $A$ having $R_{out}$ as 
the outer rectangle can be computed in $O(m)$ time, where $m$ is the number of points inside $R_{out}$.
\end{observation}

Our proposed algorithm computes an empty annulus $A^{max}_{(a,b)}$ of maximum-width within each strip $S(a,b)$, for all
such possible pairs $(a,b)$, where $a,b \in P$. Finally an annulus of maximum-width among all those annuli ($A^{max}_{(a,b)}$)
is reported.\\

\subsection{Finding an empty annulus of maximum-width in a horizontal strip}\label{finding_cand}
Consider a strip $S(a,b)$ where we are looking for an empty annulus of maximum-width from $E(a,b)$. The following approach 
presents the way to achieve our goal.\\

\begin{figure}[]
 \centering
 \includegraphics[scale=0.5]{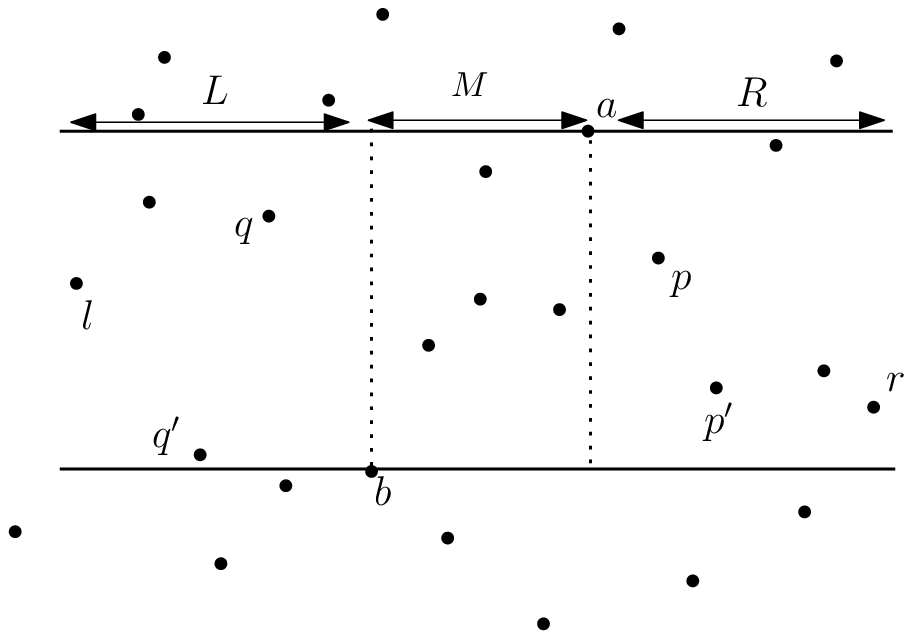}
 \caption{\small{Demonstration shows partition of Q into sets L, M, R.}}
 \label{fig:points_strip}
 \end{figure} 

Let $Q$ be the set of ordered points in increasing order (w.r.t. value of x-coordinates) including points $a$ and $b$ in 
the strip $S(a,b)$. Note that $Q$ can be determined from ordered set of points $P$ in linear time. Also let the leftmost 
and rightmost points of $Q$ be $l$ and $r$ respectively. 
Without loss of generality we have assumed that $x(a) > x(b)$ in $S(a,b)$.
To compute the elements in $E(a,b)$, we use two vertical segments $Left_{out}$ and $Right_{out}$. These two lines define 
the left and right edges of outer rectangle $R_{out}$ of an annulus $A \in E(a,b)$. 
It can be observed that $Left_{out}$ is required to sweep over the points of $Q$ on the left of $b$ and $Right_{out}$ is 
required to sweep over the points of $Q$ on the right of $a$ to generate elements of $E(a,b)$.
If any one of these segments moves to a point of $Q$ an annulus defined by the current positions of $Left_{out}$
and $Right_{out}$ is required to update and therefore the points of $Q$ are the {\sf event} points of the proposed sweep 
line algorithm. We now partition $Q$ into $3$ 
groups - $(i)$ Points starting from $x(l)$ to $x(b)$ are in set $L$, 
$(ii)$ Points inside the rectangle formed by corner points $x(a)$ and $x(b)$ in set $M$ and $(iii)$ Points starting from 
$x(a)$ to $x(r)$ in set $R$ (See Fig. \ref{fig:points_strip}). $Left_{out}$ starts sweeping from $x(b)$ and moves 
towards $x(l)$ where the event points of $Left_{out}$ are the input points in $L$. Similarly $Right_{out}$ starts sweeping
from $x(a)$ and moves towards $x(r)$ and its event points are input points in $R$.\\

Let $p$ and $q$ denote the immediate right point of $a$ and immediate left point of $b$ respectively. Also let $p'$ and
$q'$ denote the immediate right and immediate left points of $p$ and $q$ respectively. See Fig. \ref{fig:points_strip}.\\

Depending on the cardinality of $M$ we have the following three cases - $(i)$ $|M| \geq 2$, $(ii)$ $|M| = 0$, and
$(iii)$ $|M| = 1$.\\

\noindent {\bfseries Case I ($|M| \geq 2$):} We first take an initial annulus $A$ in $S(a,b)$ from which we generate other
annuli in the strip. This annulus $A$ has outer rectangle say $R_{out}$. The top-right corner and bottom-left corner of 
$R_{out}$ are at points $a$ and $b$ respectively. Construct an inner rectangle $R_{in}$ 
within this $R_{out}$ using Observation \ref{obs1}. Let $Top_{in}$, $Right_{in}$, $Bot_{in}$ and $Left_{in}$ denote the 
top, right, bottom and left edges of $R_{in}$ respectively. $W(A)$ is the width of annulus $A$.\\

\begin{lemma}\label{lem1}
Consider any annulus $K \in E(a,b)$ and assume that $W(K)$ is determined by top-, bottom or left-width of annulus $K$. If the 
right edges of outer and inner rectangles of annulus $K$ is shifted towards right to obtain another annulus $K'$, then
$W(K') \leq W(K)$.
\end{lemma}

\begin{figure}[]
  \centering
  \includegraphics[scale=0.6]{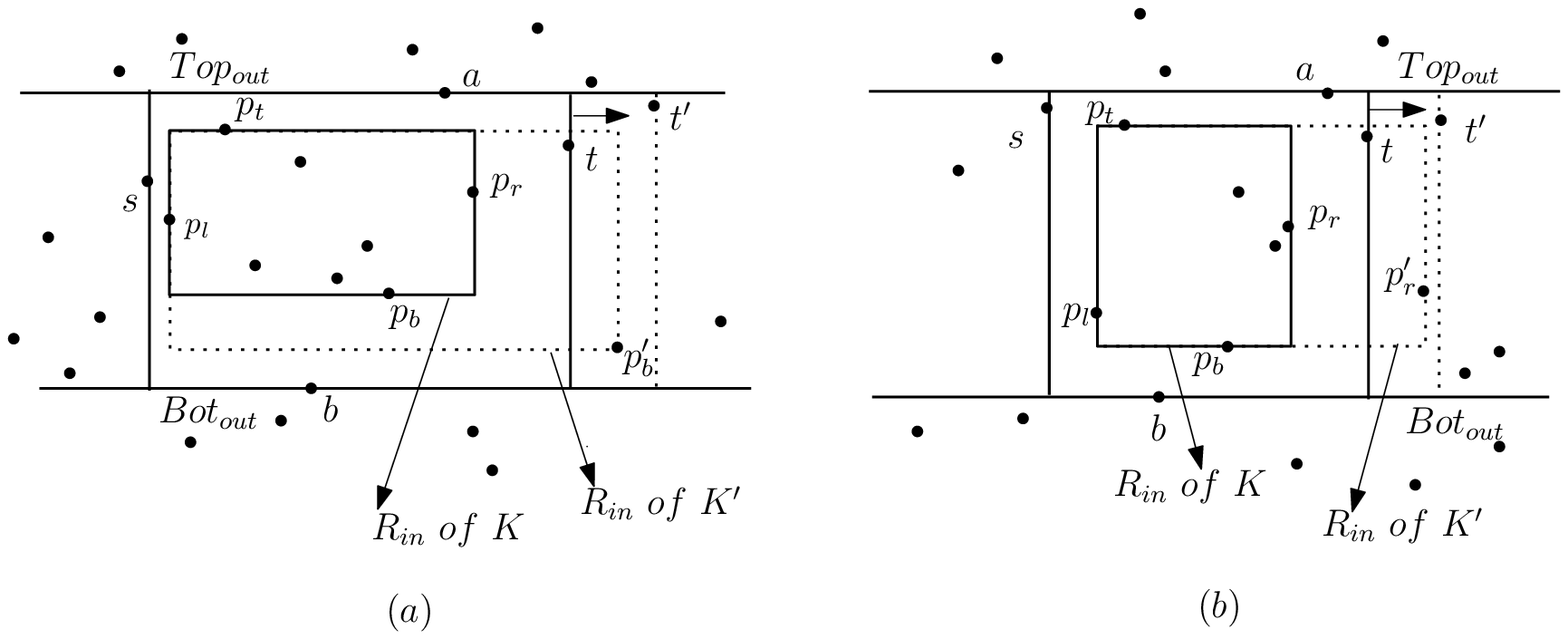}\vspace{-0.1in}
  \caption{\small{Two cases of construction of new annulus $K'$ from annulus $K$. 
  (a) Left-width of $K$ becomes $W(K)$. $K'$ is formed by shifting right edges of $R_{in}$ and $R_{out}$ of $K$. $Right_{in}$
  and $Bot_{in}$ edges of $K'$ pass through the corner point $p_b'$.
  Left width of $K'$ becomes $W(K')$ showing $W(K)= W(K')$. (b) Top width of $K$ becomes $W(K)$. $K'$ is generated from $K$ 
  and $Right_{in}$ passes through $p_r'$. Here right-width of $K'$ is $W(K')$ and shows $W(K')<W(K)$.}}\vspace{-0.1in}  
 \label{FigLemma1}
 \end{figure}
 
\begin{proof}
Let the top- ($Top_{in}$), right- ($Right_{in}$), bottom- ($Bot_{in}$), left- ($Left_{in}$) edges of inner 
rectangle ($R_{in}$) of annulus $K$ pass through the points $p_t$, $p_r$, $p_b$, and $p_l$ respectively. The $Left_{out}$ and 
$Right_{out}$ edges of $K$ pass through the points $s$ and $t$ where $s \in L$, $t \in R$ and $W(K)$ is determined by 
the left-width of $K$. See Fig.\ref{FigLemma1}(a) for an illustration. We now shift $Right_{out}$ of annulus $K$ from $t$ to 
a point $t'$ in the right where $x(t)<x(t')$ and $t' \in R$ keeping $Left_{out}$ fixed at $s$. Annulus $K'$ is constructed
where $K' \in E(a,b)$ and let us assume that $W(K')>W(K)$. Since the left edges of outer and inner rectangles of 
both $K$ and $K'$ pass through same points $s$ and $p_l$, their left-widths are equal.
If the $Bot_{in}$ edge of $K'$ is determined by a point say $p_b'$ 
such that $x(p_b') \geq x(t)$ and $p_b \in R$ then bottom-width of $K'$ is less than bottom-width of $K$. Similarly we can say 
this for top-width of $K'$. Right-width of $K'$ can be equal, greater or smaller than right-width of $K$. 
If any one of the top-, right- or bottom-widths of $K'$ is smaller than its left-width then $W(K')<W(K)$. So it contradicts
our assumption. If left-width of $K'$ determines $W(K')$ then we have $W(K')=W(K)$.\\

Now consider annulus $K$ where its top-width determines $W(K)$. Now we shift $Right_{out}$ of 
annulus $K$ from $t$ to any point $t'$ in the right where $x(t)<x(t')$ and $t, t' \in R$ and $Left_{out}$ fixed at 
$s$, $s \in L$ (See Fig.\ref{FigLemma1}(b)). 
Annulus $K'$ is formed. $Right_{in}$ of $K'$ pass through $p'_r$ where $p'_r \in R$. To achieve better solution 
i.e. $W(K')>W(K)$ we have to increase the top-width of $K'$. However this is not 
possible because the point $p_t$ will lie in the open region between $R_{in}$ and $R_{out}$ of annulus $K'$. This means
that no further sweeping of $Right_{out}$ and $Right_{in}$ of annulus $K$ is required. Using symmetry the assertion that
$W(K') \leq W(K)$ holds when $W(K)$ is determined by the bottom-width of $K$ and $K'$ is any annulus whose left edge of outer 
rectangle lies at the same position where $Left_{out}$ of $K$ lies, and right edge of outer rectangle lies to the
right of $Right_{out}$ of $K$.\qed
\end{proof}
Similarly we can prove the following result.\vspace{-0.1in}

\begin{lemma}\label{Lem2}
Assume that $K$ is an annulus in $E(a,b)$ where $W(K)$ is determined by top-, bottom- or right-width of annulus $K$. If the 
left edges of outer and inner rectangles of annulus $K$ is shifted towards left to obtain another annulus $K'$, then
$W(K') \leq W(K)$. 
\end{lemma}

\begin{figure}[h]
 \centering
 \includegraphics[scale=0.6]{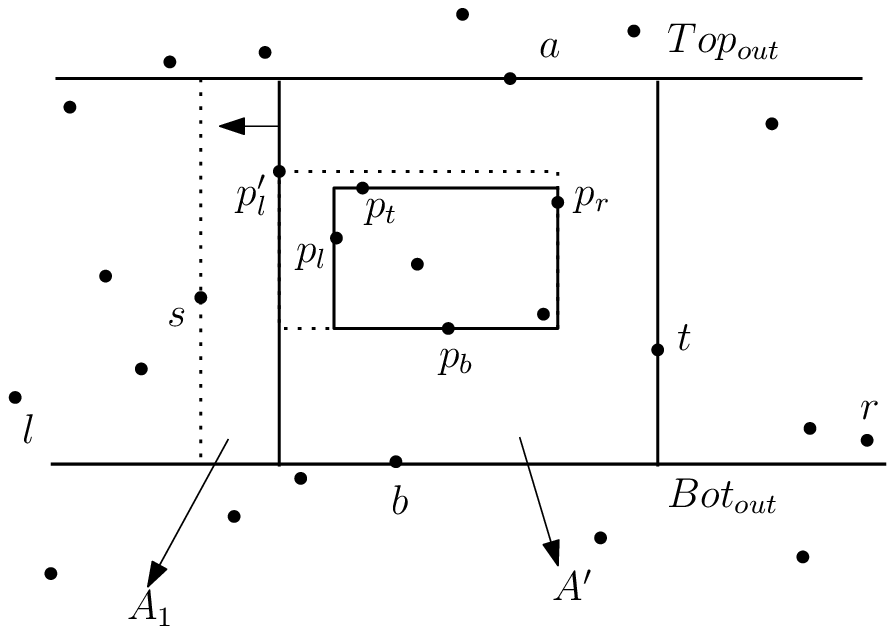}\vspace{-0.1in}
 \caption{\small{Demonstration shows construction of $A_1$ from $A'$. Left-width of $A'$ becomes $W(A')$. $Top_{in}$, 
 $Right_{in}$, 
$Bot_{in}$ and $Left_{in}$ edges of $A'$ pass through the points $p_t$, $p_r$, $p_b$, $p_l$ respectively. $Right_{out}$ of
both $A'$ and $A_1$ pass through $t$. $Left_{out}$ of $A'$ and $A_1$ pass through $p'_l$ and $s$.}}\vspace{-0.1in}
 \label{fig:shift}
 \end{figure} 

Algorithm \ref{alg:algostrip} is based on the computation of a new annulus from an
initial configuration. Let $A'$ be a given annulus in $E(a,b)$. See Fig.\ref{fig:shift} as illustration. Depending on $W(A')$
we {\sf shift}
$Left_{out}$ of $A'$ from $p'_l$ to the next event point in left $s$, where $p'_l, s \in L$. Therefore a new outer rectangle 
is formed. Since $p'_l$ lies in the open region between this new outer rectangle and $R_{in}$ of $A'$, $p'_l$ is compared
with the points $p_t$, $p_b$ and $p_l$. We thus create new annulus $A_1$. Note that this operation requires constant
time. Now we describe Algorithm \ref{alg:algostrip} to compute the set $E(a,b)$ for Case I. In each step our algorithm keeps 
information about the best solution computed so far. Let the best solution in
$S(a,b)$ is stored in $W(A^{max}_{(a,b)})$ where $A^{max}_{(a,b)}$ is an maximum-width annulus in the strip.
%
%

\begin{algorithm}[]
\KwIn{Annulus $A$ whose outer rectangle is defined by two opposite corner points $a$ and $b$ and its $Left_{out}$ and 
$Right_{out}$ passes through $b$ and $a$. $L, M, R$ are set of ordered points in $S(a,b)$ in increasing order (w.r.t.
the value of x-coordinates) where $L, M, R$ are obtained from $Q$ in $S(a,b)$.}
\KwOut{The width $W(A^{max}_{(a,b)})$ of an empty annulus $A^{max}_{(a,b)}$ of maximum-width in $S(a,b)$.}
$W(A^{max}_{(a,b)}) \gets W(A)$. \\
\While{$Left_{out}$ {\em and} $Right_{out}$  {\em do not pass through} $l$  {\em and} $r$  {\em respectively}} {
    \If { {\em top-width (or bottom-width) of $A$} {\em determines} $W(A)$}    
        {\If {$W(A) > W(A^{max}_{(a,b)})$}
        {$W(A^{max}_{(a,b)}) \gets W(A)$}
       {\bf Exit} }

    \If { {\em left-width of} $A$  {\em determines} $W(A)$  \& $Left_{out}$  {\em passes through} $l$} 
        {\If {$W(A) > W(A^{max}_{(a,b)})$}
        {$W(A^{max}_{(a,b)}) \gets W(A)$}
              {\bf Exit} }

    \If {{\em right-width of} $A$  {\em determines} $W(A)$ \& $Right_{out}$  {\em passes through} $r$} 
        {\If{$W(A) > W(A^{max}_{(a,b)})$}
        {$W(A^{max}_{(a,b)}) \gets W(A)$}
            {\bf Exit} }

    \If {{\em left-width of} $A$  {\em determines} $W(A)$  {\em then shift} $Left_{out}$  {\em of} $A$  {\em to the next event point in left. Let} $A'$  {\em is
    the new annulus formed}}
        {\If{$W(A') > W(A^{max}_{(a,b)})$}
          {$W(A^{max}_{(a,b)}) \gets W(A')$}
          {$A \gets A'$ (Ref. Lemma \ref{lem1})}}
          
    \If {{\em right-width of} $A$  {\em determines} $W(A)$ {\em then
     shift} $Right_{out}$  {\em of} $A$  {\em to the next event point in right. Let} $A'$ {\em be the new annulus formed}}
        {\If{$W(A') > W(A^{max}_{(a,b)})$}
         {$W(A^{max}_{(a,b)}) \gets W(A')$}
         {$A \gets A'$ (Ref. Lemma \ref{Lem2})}}

    }
  Return $W(A^{max}_{(a,b)})$. 
\caption{Algorithm for computing an empty annulus of maximum-width in $S(a,b)$.}
\label{alg:algostrip}
\end{algorithm}

Note that if top- (or bottom) width becomes width of an annulus and is equal to its left- (or right) width we consider
its top- (or bottom) width as its width (Followed from Lemma \ref{lem1}). Also if any annulus have width from its left-width 
and right-width simultaneously then we consider any one of them as its width and proceed accordingly.

\begin{figure}
 \centering
 \includegraphics[scale=0.6]{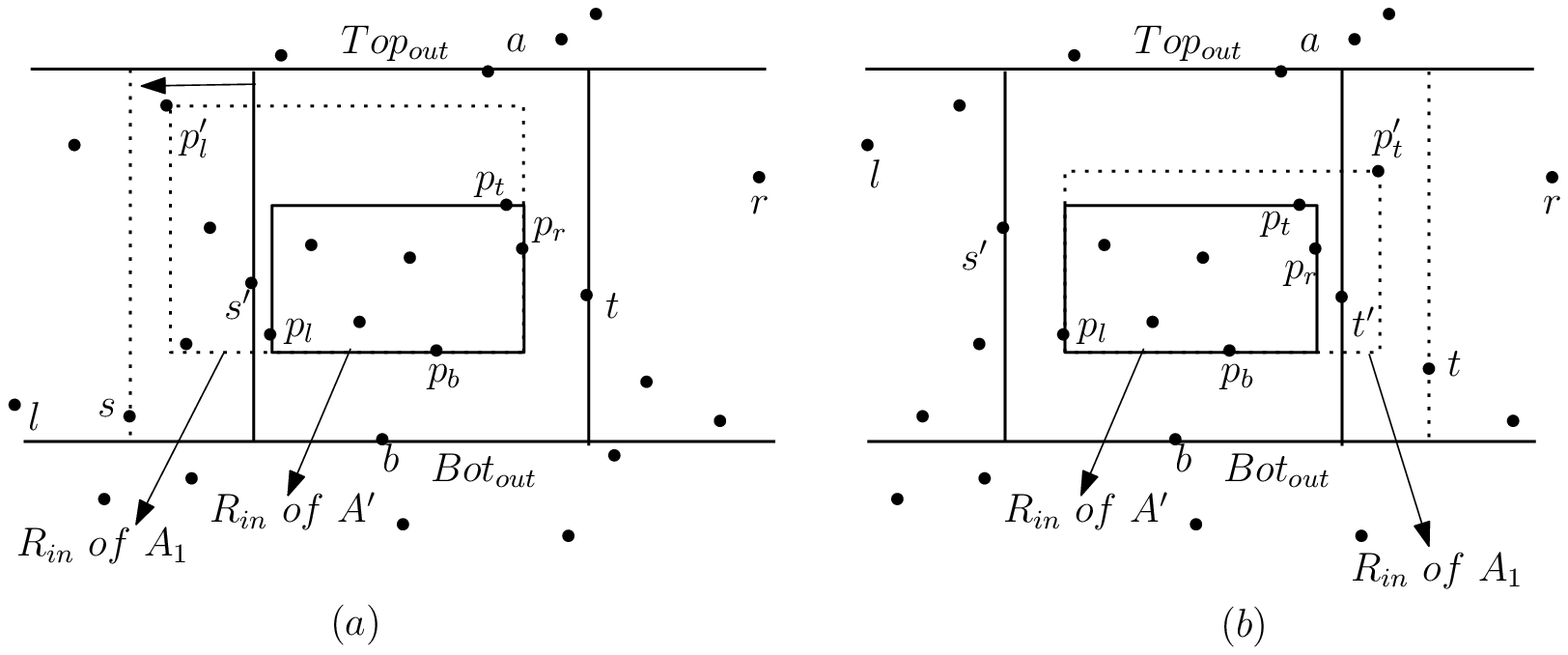}\vspace{-0.1in}
 \caption{\small{(a)$W(A')$ and $W(A_1)$ are determined by left-widths of $A'$ and $A_1$. (b)$W(A')$
 and $W(A_1)$ are determined by right-widths of $A'$ and $A_1$.}} \vspace{-0.1in} 
 \label{lem2_fig}
 \end{figure}
 
In Algorithm \ref{alg:algostrip} all elements of $E(a,b)$ are not computed. It starts with the initial configuration of 
annulus $A$. Depending on $W(A)$ we shift either its left edge or right edge of outer rectangle. Assume that 
left-width of $A$ determines $W(A)$. Now consider $A'$ and $A_1$ are two annuli computed in Case I where left-widths of 
$A'$ and $A_1$ determines $W(A')$ and $W(A_1)$ respectively. See Fig.\ref{lem2_fig}(a) as an illustration.
$Left_{out}$ of $A'$ and $A_1$ pass through $s'$ and $s$, and $s', s \in L$. $Right_{out}$ of $A'$ and $A_1$ 
pass through point $t$, $t \in R$. 
Let $A_L$ be the set of all annuli whose left edges ($Left_{out}$) of outer rectangle pass through any point between
$s'$ and $s$. The right edge of outer rectangle of any annulus in $A_L$ is fixed at $t$.  
If we shift the $Right_{out}$ of $A'$ to any point $t'$ such that $x(t') > x(t)$ and $t' \in R$ then the annulus
that will be created have width either less or equal to $W(A')$ (Ref. Lemma \ref{lem1}). This fact is true for all
annuli in $A_L$. This means that there is no requirement to generate all those annuli whose $Left_{out}$ pass through any point
from $s'$ to $s$ and $Right_{out}$ passes through any point in the right of $t$. It may happen that $Left_{out}$ of $A_1$ 
reaches $l$ and left-width of $A_1$ determines $W(A_1)$ then our algorithm terminates and reports the best solution in $S(a,b)$.
If right-width of $A_1$ determines $W(A_1)$ then we shift the $Right_{out}$ of $A_1$ and compute annuli further. Now assume 
that $W(A')$ and $W(A_1)$ are determined by right-widths of $A'$ and $A_1$. $Right_{out}$ of $A'$ and $A_1$ pass through 
$t'$ and $t$ and $t', t \in R$ and their $Left_{out}$ is fixed at $s'$, $s' \in L$. Let $A_R$ be the set of all annuli whose
right edges of outer rectangle pass through any point between $t'$ and $t$ and left edges of outer rectangle fixed at $s'$ (See 
Fig.\ref{lem2_fig}(b)). In a similar way we can say that there is no requirement to shift the left edge of outer rectangle of
any annulus in $A_R$.\\

We now consider the case when $|M| = 0$.

\noindent {\bfseries Case II($|M| = 0$):} We use Algorithm \ref{alg:algostrip} to generate annuli of $E(a,b)$. It requires an 
initial configuration. We need at least two points to create an inner rectangle. These two points to form inner rectangle can 
lie in the left side of both $a$ and $b$, in the right side of both $a$ and $b$, or one in the right side of $a$ and other in
the left side of $b$. Thus we need three initial configurations of the annuli from which 
we can generate other annuli in $E(a,b)$. They are as follows.\\

$(i)$ Outer rectangle formed by point $a$ on the top right corner, $b$ in the bottom and $Left_{out}$ passing through the 
point $q_1$ where $q_1$ is immediate left point of $q'$ (See Fig. \ref{case2}(a)). Two points $q$ and $q'$ lie on the 
two opposite corner of the inner rectangle. We name this annulus as $A_1$.\\

$(ii)$ Outer rectangle formed by point $b$ on the lower left corner, $a$ on the above and $Right_{out}$ passing through the 
point $p_1$ where $p_1$ is immediate right point of $p'$ (See Fig. \ref{case2}(b)). Two points $p$ and $p'$ lie on the 
two opposite corner of the inner rectangle. We name this annulus as $A_2$.\\

$(iii)$ Outer rectangle formed by point $q'$ on the left, $a$ on the above, $p'$ on the right and $b$ lying at bottom. The 
inner rectangle is formed by two opposite corner points $p$ and $q$. Say this annulus $A_3$ (See Fig. \ref{case2}(c)). 
%

 \begin{figure}[h]
 \centering
 \includegraphics[scale=0.6]{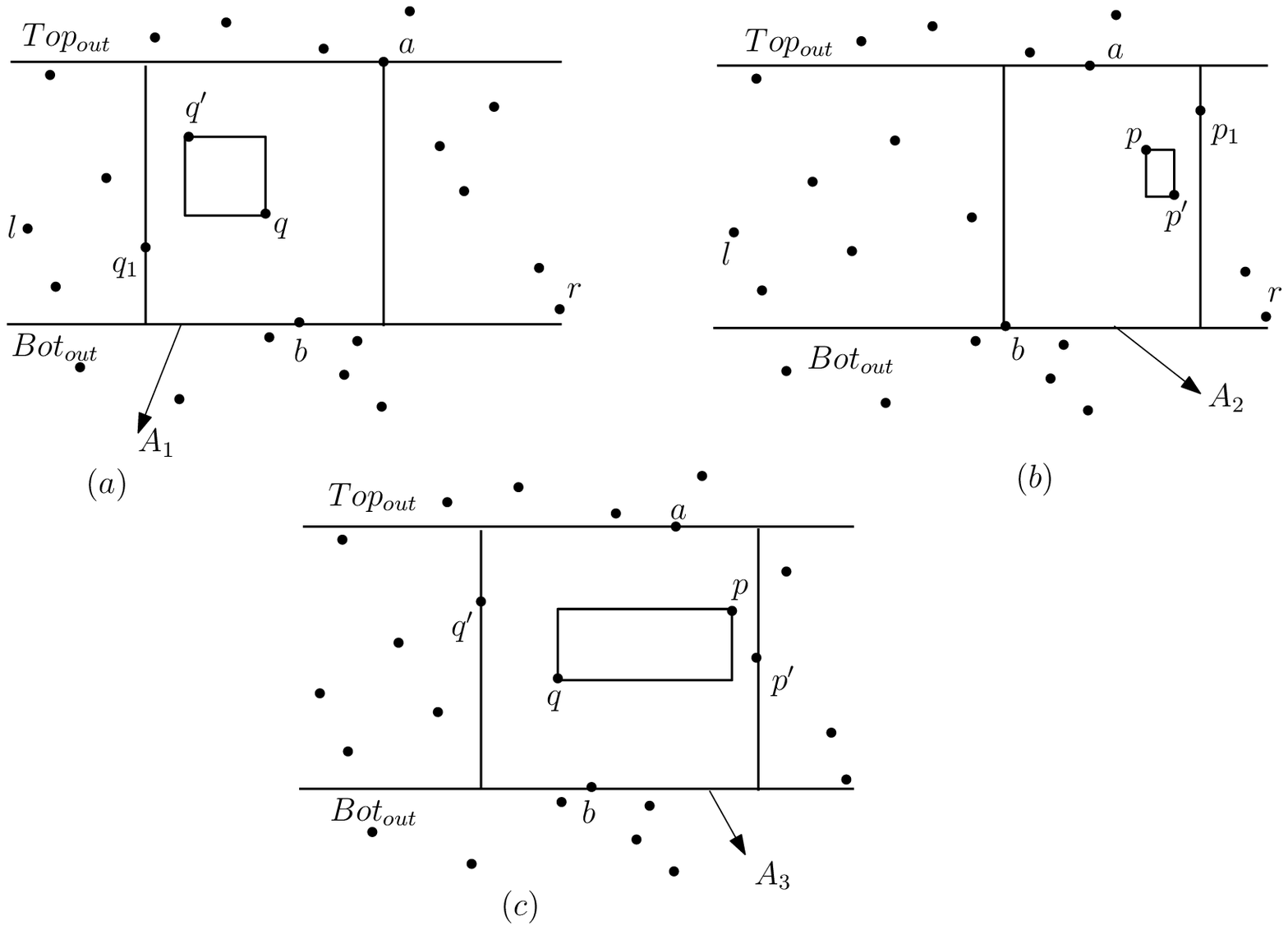}
 \caption{Demonstration shows $A_1, A_2, A_3$.}  
 \label{case2}
 \end{figure}

%
%
For each initial configuration we invoke Algorithm \ref{alg:algostrip}. We compare the solutions obtained from them
and finally report an empty annulus of maximum-width in $S(a,b)$.\\
%

\noindent {\bfseries Case III($|M| = 1$):} A single point, say $z$ is present inside the outer rectangle formed by two opposite 
corner points $a$ and $b$ in $S(a,b)$. Algorithm \ref{alg:algostrip} requires an initial configuration to start with. We need 
at least two points to create an inner rectangle. One of them is $z$ and the other point can lie either in the left side or 
right side of $z$. Therefore we form two initial configurations of annuli to compute other annuli in $E(a,b)$.\\

$(i)$ Outer rectangle formed by point $a$ on the top right corner, $b$ in the bottom and $Left_{out}$ passing through the 
point $q'$ (See Fig. \ref{case3}(a)). Opposite corner points $q$ and $z$ form inner rectangle. This annulus is $A_1$.\\

$(ii)$ Outer rectangle formed by point $b$ on the lower left corner, $a$ on the above and $Right_{out}$ passing through the
point $p'$. Here $p$ and $z$ are used to form inner rectangle. See Figure \ref{case3}(b). Let this annulus be $A_2$.\\
%

\begin{figure}[h]
 \centering
 \includegraphics[scale=0.55]{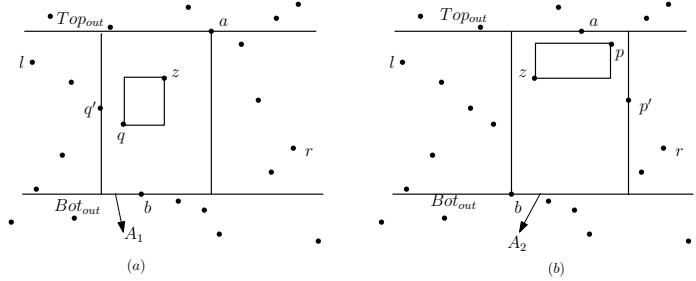}
 \caption{Demonstration shows $A_1, A_2$.}  
 \label{case3}
 \end{figure}
%
We invoke Algorithm \ref{alg:algostrip} separately on $A_1$, $A_2$ and report an empty annulus of maximum-width in $S(a,b)$.
As stated in Case I, we do not compute all elements of $E(a,b)$ for Case II and Case III. We report an empty annulus of 
maximum-width in $S(a,b)$ from those annuli which are computed in Case II (resp. Case III). \\
%
%
Now we have the following result.\\
\begin{theorem}\label{timecom}
For a given set of points $P = \{p_1, \ldots, p_n\}$ in $\mathbb R^{2}$, an empty
annulus of maximum-width can be computed in $O(n^3)$ time using $O(n)$ space.
\end{theorem}

\begin{proof}
Note that the number of horizontal strips formed by any two points of $P$ is $O(n^2)$.
Algorithm \ref{alg:algostrip} requires $O(m)$ time where $m$ ($\leq n$) is the number of input points in any such strip. 
Thus the result follows. \qed
\end{proof}

In the above axis-parallel empty rectangular annulus problem, the inner rectangle forming such an empty annulus 
always have {\sf non-zero} area. However if $R_{in}$ reduces to a single point $p \in P$ then we have following result.

\begin{corollary}\label{coro111}
Given a set $P$ of $n$ points on $\mathbb R^{2}$, an axis-parallel empty rectangular annulus $A$ of maximum-width can be
computed in $O(n\log n)$ time using $O(n)$ space when the inner rectangle of the annulus $A$ reduces to a single point $p \in P$.
\end{corollary}

\begin{proof} One can construct the voronoi diagram for point set $P$ in $O(n \log n)$ time using $O(n)$ size data 
structure~\cite{bcko-cgaa-08}. For any query point $q \in P$, a nearest point of $q$ among the points from $P$ can be
computed in $O(\log n)$ time.
Therefore the computation of nearest input points for all points of $P$ requires $O(n \log n)$ time.
Hence the result follows.  \qed
\end{proof}

\section{Conclusion and discussion}\label{conclusion}
In the annulus problem studied by Mukherjee et al.~\cite{jmkd-mwra-2012}, the outer rectangle of an annulus of minimum-width
which encloses $P$ must be the minimum enclosing rectangle enclosing $P$ where $P$ is the set of $n$ input points 
in $\mathbb R^{2}$.
However in the empty axis-parallel rectangular annulus problem of maximum-width, the number of potential outer rectangles 
is $O(n^4)$.
This observation implies that an $O(n^5)$ algorithm is trivial to find an empty rectangular axis-parallel annulus of maximum 
width. Therefore the proposed $O(n^3)$ time algorithm to solve the maximum-width empty annulus problem is a non-trivial one. 
Note that we didn't give any lower bound for this problem but proposed 
$O(n\log n)$ time algorithm in Corollary~\ref{coro111} to solve the problem for a particular case. In this context, it would 
be interesting to give a sub-quadratic algorithm or to prove the problem $O(n^2)$-hard.
Note that for each empty
rectangular annulus problem discussed so far, the orientation is fixed. In future it remains as a challenge to solve this 
problem where the annuli are of {\sf arbitrary orientation}.

\section{Acknowledgements}
This work is supported by Project (Ref. No. $248 (19)~2014~$R $\&$ D II$~\/ 1045$) from The National Board for Higher 
Mathematics (NBHM), Government of India awarded to P. Mahapatra where Arpita Baral is a research scholar under this Project.

\bibliographystyle{plain}
\bibliography{ptdom-April-2015}
\end{document}